\documentclass{article}

\usepackage{amssymb,amsthm,amsmath}

\newcommand{\ZZ}{\mathbb{Z}}

\newcommand{\FF}{\mathbb{F}}

\newtheorem{theorem}{Theorem}[section]
\newtheorem{lemma}[theorem]{Lemma}
\newtheorem{proposition}[theorem]{Proposition}
\newtheorem{corollary}[theorem]{Corollary}

\theoremstyle{definition}

\newtheorem{remark}[theorem]{Remark}
\usepackage[square,numbers]{natbib}
\usepackage{hyperref}
\usepackage{breakurl}
\usepackage{url}
\hypersetup{pdfpagemode=none,colorlinks,citecolor=blue,filecolor=black,linkcolor=black,urlcolor=blue}
\newcommand{\vct}[1]{\mathbf{#1}}
\newcommand{\nth}{^{\text{th}}}
\DeclareMathOperator{\moddec}{mod}
\renewcommand{\mod}[1]{\,(\moddec #1)}
\usepackage{enumitem}
\usepackage{nameref}
\makeatletter
\newcommand{\clabel}[2]{\protected@write \@auxout {}{\string \newlabel {#1}{{#2}{\thepage}{#2}{#1}{}} }\hypertarget{#1}{}}
\makeatother
\newcommand{\lb}{\allowbreak}
\usepackage{etoolbox}
\newcommand{\breaklist}[2][,\lb]{\def\nextitem{\def\nextitem{#1}}\renewcommand*{\do}[1]{\nextitem{##1}}\docsvlist{#2}}
\DeclareMathOperator{\cirdec}{circ}		
	
\newcommand{\cir}[1]{\cirdec(\breaklist{#1})}

\DeclareMathOperator{\rankdec}{rank}
\newcommand{\rank}[1]{\rankdec(\breaklist{#1})}	
\newcommand{\floor}[1]{\lfloor#1\rfloor}
\usepackage{empheq}

\DeclareMathOperator{\autdec}{Aut}
\newcommand{\aut}[1]{\autdec(#1)}
\usepackage{booktabs}
\usepackage{upgreek}
\newcommand{\vctg}[1]{\boldsymbol{#1}}
\makeatletter
\renewcommand*\env@matrix[1][*\c@MaxMatrixCols c]{\hskip -\arraycolsep\let\@ifnextchar\new@ifnextchar\array{#1}}
\makeatother
\usepackage[font=footnotesize,labelfont=bf,textfont=normalfont,margin=0cm]{caption}
\usepackage{adjustbox}
\newcommand{\pcir}[2]{\cirdec_{#1}(\breaklist{#2})}
\providecommand{\keywords}[1]{\small\textit{Keywords}: #1}
\providecommand{\msc}[1]{\small\textit{2020 MSC}: #1}

\title{New binary self-dual codes of lengths 56, 62, 78, 92 and 94 from a bordered construction}

\author{J. Gildea, A. Korban and A. M. Roberts\\
Department of Physical, Mathematical and Engineering Sciences\\
University of Chester\\
Exton Park\\
Chester CH1 4AR\\
United Kingdom\\
{}\\
A. Tylyshchak\\
Department of Algebra\\
Uzhgorod National University\\
Uzhgorod\\
Ukraine
}
\date{}

\begin{document}

\maketitle


\msc{94B05, 15B10, 15B33}

\let\thefootnote\relax\footnote{E-mail addresses: \href{mailto:j.gildea@chester.ac.uk}{j.gildea@chester.ac.uk} (J. Gildea), \href{mailto:adrian3@windowslive.com}{adrian3@windowslive.com} (A. Korban), \href{mailto:adammichaelroberts@outlook.com}{adammichaelroberts@outlook.com} (A. M. Roberts),
\href{mailto:}{alxtlk@bigmir.net} (A. Tylyshchak)
}

\begin{abstract}
In this paper, we present a new bordered construction for self-dual codes which employs $\lambda$-circulant matrices. We give the necessary conditions for our construction to produce self-dual codes over a finite commutative Frobenius ring of characteristic 2. Moreover, using our bordered construction together with the well-known building-up and neighbour methods, we construct many binary self-dual codes of lengths 56, 62, 78, 92 and 94 with parameters in their weight enumerators that were not known in the literature before. 
\end{abstract}

\section{Introduction}

The study of self-dual codes over finite fields and rings is an active area of research in coding theory. This is mainly due to their connections to other areas in mathematics such as combinatorics, design theory and number theory. In \cite{R-218}, it is shown that one can produce interesting designs using self-dual codes over fields. In \cite{R-219}, it is shown that one of the most powerful techniques for producing optimal unimodular lattices uses self-dual codes over rings. Moreover, as seen in \cite{R-220}, the well-known proof of the non-existence of the projective plane of order 10 used the theory of binary self-dual codes.

Finding methods for constructing self-dual codes, classifying self-dual codes and determining the largest minimum distance among all self-dual codes are open problems in coding theory. Many researchers have employed different techniques to construct self-dual codes. For example, in \cite{R-008}, the authors establish a strong connection between group rings and self-dual codes. They prove that a group ring element corresponds to a self-dual code if and only if it is a unitary unit and they use this fact to construct many new extremal binary self-dual codes of length 68. In \cite{R-031,R-193,R-086,R-125,R-137}, the authors employ a particular technique which guarantees that the constructed self-dual codes have a fixed automorphism of odd prime order.

A powerful method for finding new binary self-dual codes is to consider bordered matrix constructions. There exist several papers in which such constructions are used to obtain binary self-dual codes of different lengths with new parameters in their weight enumerators. Examples of bordered constructions can be found in \cite{R-108,R-043,R-098}. In this work, we continue in this direction and we present a new bordered matrix construction that we employ to obtain many binary self-dual codes of lengths 56, 62, 78, 92 and 94 with weight enumerators that were not known in the literature before. Moreover, we present the necessary conditions that the construction needs to satisfy in order to produce a self-dual code over a finite commutative Frobenius ring of characteristic 2.

The paper is organised as follows. In Section 2, we give the standard definitions and results on self-dual codes, the alphabets we use, Gray maps and some special matrices that we use in this work. In Section 3, we present the main construction and show under what conditions it produces self-dual codes over a finite commutative Frobenius ring of characteristic 2. In Section 4, we present our computational results. Namely, we tabulate all new binary self-dual codes that we obtain by a direct application of our main construction together with the building-up and neighbour constructions. We finish with concluding remarks and directions for possible future research.

\section{Preliminaries}

\subsection{Self-Dual Codes}

Let $R$ be a commutative Frobenius ring (see \cite{B-014} for a full description of Frobenius rings and codes over Frobenius rings). Throughout this work, we always assume $R$ has unity. A code $\mathcal{C}$ of length $n$ over $R$ is a subset of $R^n$ whose elements are called codewords. If $\mathcal{C}$ is a submodule of $R^n$, then we say that $\mathcal{C}$ is linear. Let $\mathbf{x},\mathbf{y}\in R^n$ where $\mathbf{x}=(x_1,x_2,\dots,x_n)$ and $\mathbf{y}=(y_1,y_2,\dots,y_n)$. The (Euclidean) dual $\mathcal{C}^{\bot}$ of $\mathcal{C}$ is given by
	\begin{equation*}
	\mathcal{C}^{\bot}=\{\mathbf{x}\in R^n: \langle\mathbf{x},\mathbf{y}\rangle=0,\forall\mathbf{y}\in\mathcal{C}\},
	\end{equation*}	
where $\langle\cdot,\cdot\rangle$ denotes the Euclidean inner product defined by
	\begin{equation*}
	\langle\mathbf{x},\mathbf{y}\rangle=\sum_{i=1}^nx_iy_i
	\end{equation*}
and we say that $\mathcal{C}$ is self-orthogonal if $\mathcal{C}\subseteq \mathcal{C}^\perp$ and self-dual if $\mathcal{C}=\mathcal{C}^{\bot}$.

An upper bound on the minimum (Hamming) distance of a doubly-even (Type II) binary self-dual code was given in \cite{R-116} and likewise for a singly-even (Type I) binary self-dual code in \cite{R-115}. Let $d_{\text{I}}(n)$ and $d_{\text{II}}(n)$ be the minimum distance of a Type I and Type II binary self-dual code of length $n$, respectively. Then
	\begin{equation*}
	d_{\text{II}}(n)\leq 4\floor{n/24}+4
	\end{equation*}
and
	\begin{equation*}
	d_{\text{I}}(n)\leq
	\begin{cases}
	4\floor{n/24}+2,& \text{if }n\equiv 0\pmod{24},\\
	4\floor{n/24}+4,& \text{if }n\not\equiv 22\pmod{24},\\
	4\floor{n/24}+6,& \text{if }n\equiv 22\pmod{24}.
	\end{cases}
	\end{equation*}

A self-dual code whose minimum distance meets its corresponding bound is called \textit{extremal}. A self-dual code with the highest minimum distance for its length is said to be \textit{optimal}. Extremal codes are necessarily optimal but optimal codes are not necessarily extremal. A \textit{best known} self-dual code is a self-dual code with the highest known minimum distance for its length.

\subsection{Alphabets}

In this paper, we consider the alphabets $\FF_2$, $\FF_2+u\FF_2$, $\FF_2+u\FF_2+v\FF_2+uv\FF_2$ and $\FF_4+u\FF_4$.

Define
	\begin{equation*}
	\FF_2+u\FF_2=\{a+bu:a,b\in\FF_2,u^2=0\},
	\end{equation*}
then $\FF_2+u\FF_2$ is a commutative ring of order 4 and characteristic 2 such that $\FF_2+u\FF_2\cong\FF_2[u]/\langle u^2\rangle$. 

Define
	\begin{equation*}
	\FF_2+u\FF_2+v\FF_2+uv\FF_2=\{a+bu+cv+duv:a,b,c,d\in\FF_2,u^2=v^2=uv+vu=0\},
	\end{equation*}
then $\FF_2+u\FF_2+v\FF_2+uv\FF_2$ is a commutative ring of order 16 and characteristic 2 such that $\FF_2+u\FF_2+v\FF_2+uv\FF_2\cong\FF_2[u,v]/\langle u^2,v^2,uv+vu\rangle$. Note that $\FF_2+u\FF_2+v\FF_2+uv\FF_2$ can be viewed as an extension of $\FF_2+u\FF_2$ and so we can also define
	\begin{equation*}
	\FF_2+u\FF_2+v\FF_2+uv\FF_2=\{a+bv: a,b\in \FF_2+u\FF_2,v^2=0\}.
	\end{equation*}

We define $\FF_4\cong\FF_2[\omega]/\langle \omega^2+\omega+1\rangle$ so that
	\begin{equation*}
	\FF_4=\{a{\omega}+b(1+\omega): a,b\in\FF_2,\omega^2+\omega+1=0\}.
	\end{equation*}

Define
	\begin{equation*}
	\FF_4+u\FF_4=\{a+bu: a,b\in\FF_4,u^2=0\},
	\end{equation*}
then $\FF_4+u\FF_4$ is a commutative ring of order 16 and characteristic 2 such that $\FF_4+u\FF_4\cong\FF_4[u]/\langle u^2\rangle\cong\FF_2[\omega,u]/\langle \omega^2+\omega+1,u^2,\omega u+u\omega\rangle$. Note that $\FF_4+u\FF_4$ can be viewed as an extension of $\FF_2+u\FF_2$ and so we can also define
	\begin{equation*}
	\FF_4+u\FF_4=\{a\omega+b(1+\omega): a,b\in \FF_2+u\FF_2,\omega^2+\omega+1=0\}.
	\end{equation*}

We recall the following Gray maps from \cite{R-117,R-118,R-119,R-129}:
	\begin{align*}
	\varphi_{\FF_2+u\FF_2}&:(\FF_2+u\FF_2)^n\to\FF_2^{2n}\\
        &\quad a+bu\mapsto(b,a+b),\,a,b\in\FF_2^{n},\\[6pt]
	\varphi_{\FF_4+u\FF_4}&:(\FF_4+u\FF_4)^n\to\FF_4^{2n}\\
		&\quad a+bu\mapsto(b,a+b),\,a,b\in\FF_4^{n},\\[6pt]
	\phi_{\FF_2+u\FF_2+v\FF_2+uv\FF_2}&:(\FF_2+u\FF_2+v\FF_2+uv\FF_2)^n\to (\FF_2+u\FF_2)^{2n}\\
		&\quad a+bv\mapsto(b,a+b),\,a,b\in (\FF_2+u\FF_2)^{n},\\[6pt]
	\psi_{\FF_4}&:\FF_4^n\to\FF_2^{2n}\\
		&\quad a\omega+b(1+\omega)\mapsto(a,b),\,a,b\in\FF_2^n,\\[6pt]
	\psi_{\FF_4+u\FF_4}&:(\FF_4+u\FF_4)^n\to (\FF_2+u\FF_2)^{2n}\\
		&\quad a\omega+b(1+\omega)\mapsto(a,b),\,a,b\in (\FF_2+u\FF_2)^n
	\end{align*}
and we note that these Gray maps preserve orthogonality in their respective alphabets (see \cite{R-119} for details). If $\mathcal{C}\subseteq (\FF_4+u\FF_4)^n$, then the binary codes $\varphi_{\FF_2+u\FF_2}\circ\psi_{\FF_4+u\FF_4}(\mathcal{C})$ and $\psi_{\FF_4}\circ\varphi_{\FF_4+u\FF_4}(\mathcal{C})$ are equivalent to each
other (see \cite{R-118,R-119} for details). The Lee weight of a codeword is defined to be the Hamming weight of its binary image under any of the previously mentioned compositions of maps. A self-dual code in $R^n$ where $R$ is equipped with a Gray map to the binary Hamming space is said to be of Type II if the Lee weights of all codewords are multiples of 4, otherwise it is said to be of Type I.
	\begin{proposition}\label{proposition-1}\textup{(\cite{R-119})}
	Let $\mathcal{C}$ be a code over $\FF_4+u\FF_4$. If $\mathcal{C}$ is self-orthogonal, then $\psi_{\FF_4+u\FF_4}(\mathcal{C})$ and $\varphi_{\FF_4+u\FF_4}(\mathcal{C})$ are self-orthogonal. The code $\mathcal{C}$ is a Type I (resp. Type II) code over $\FF_4+u\FF_4$ if and only if $\varphi_{\FF_4+u\FF_4}(\mathcal{C})$ is a Type I (resp. Type II) $\FF_4$-code if and only if $\psi_{\FF_4+u\FF_4}(\mathcal{C})$ is a Type I (resp. Type II) $(\FF_2+u\FF_2)$-code. Furthermore, the minimum Lee weight of $\mathcal{C}$ is the same as the minimum Lee weight of $\psi_{\FF_4+u\FF_4}(\mathcal{C})$ and $\varphi_{\FF_4+u\FF_4}(\mathcal{C})$.
	\end{proposition}

The next corollary follows immediately from Proposition \ref{proposition-1}.
	\begin{corollary}
	Let $\mathcal{C}$ be a self-dual code over $\FF_4+u\FF_4$ of length $n$ and minimum Lee distance $d$. Then $\varphi_{\FF_2+u\FF_2}\circ\psi_{\FF_4+u\FF_4}(\mathcal{C})$ is a binary self-dual $[4n,2n,d]$ code. Moreover, the Lee weight enumerator of $\mathcal{C}$ is equal to the Hamming weight enumerator of $\varphi_{\FF_2+u\FF_2} \circ \psi_{\FF_4+u\FF_4}(\mathcal{C})$. If $\mathcal{C}$ is a Type I (resp. Type II) code, then $\varphi_{\FF_2+u\FF_2}\circ\psi_{\FF_4+u\FF_4}(\mathcal{C})$ is a Type I (resp. Type II) code.
	\end{corollary}

We also have the following proposition from \cite{R-130}:
	\begin{proposition}\textup{(\cite{R-130})}
	Let $\mathcal{C}$ be a self-dual code over $\FF_2+u\FF_2+v\FF_2+uv\FF_2$ of length $n$ and minimum Lee distance $d$. Then $\varphi_{\FF_2+u\FF_2}\circ\phi_{\FF_2+u\FF_2+v\FF_2+uv\FF_2}(\mathcal{C})$ is a binary $[4n,2n,d]$ self-dual code. Moreover, the Lee weight enumerator of $\mathcal{C}$ is equal to the Hamming weight enumerator of $\varphi_{\FF_2+u\FF_2}\circ\phi_{\FF_2+u\FF_2+v\FF_2+uv\FF_2}(\mathcal{C})$. If $\mathcal{C}$ is a Type I (resp. Type II) code, then $\varphi_{\FF_2+u\FF_2}\circ\phi_{\FF_2+u\FF_2+v\FF_2+uv\FF_2}(\mathcal{C})$ is a Type I (resp. Type II) code.
	\end{proposition}

\subsection{Special Matrices}

We now define and discuss the properties of some special matrices which we use in our work. Let $\vct{a}=(a_0,a_1,\ldots,a_{n-1})\in R^n$ where $R$ is a commutative ring and let
	\begin{equation*}
	A=\begin{pmatrix}
	a_0 & a_1 & a_2 & \cdots & a_{n-1}\\
	\lambda a_{n-1} & a_0 & a_1 & \cdots & a_{n-2}\\
	\lambda a_{n-2} & \lambda a_{n-1} & a_0 & \cdots & a_{n-3}\\
	\vdots & \vdots & \vdots & \ddots & \vdots\\
	\lambda a_1 & \lambda a_2 & \lambda a_3 & \cdots & a_0
	\end{pmatrix},
	\end{equation*}
where $\lambda\in R$. Then $A$ is called the $\lambda$-circulant matrix generated by $\vct{a}$, denoted by $A=\pcir{\lambda}{\vct{a}}$. If $\lambda=1$, then $A$ is called the circulant matrix generated by $\vct{a}$ and is more simply denoted by $A=\cir{\vct{a}}$. If we define the matrix
	\begin{equation*}
	P_{\lambda}=\begin{pmatrix}
	\vct{0} & I_{n-1}\\
	\lambda & \vct{0}
	\end{pmatrix},
	\end{equation*}
then it follows that $A=\sum_{i=0}^{n-1}a_iP_{\lambda}^i$. Clearly, the sum of any two $\lambda$-circulant matrices is also a $\lambda$-circulant matrix. If $B=\pcir{\lambda}{\vct{b}}$ where $\vct{b}=(b_0,b_1,\ldots,b_{n-1})\in R^n$, then $AB=\sum_{i=0}^{n-1}\sum_{j=0}^{n-1}a_ib_jP_{\lambda}^{i+j}$. Since $P_{\lambda}^n=\lambda I_n$ there exist $c_k\in R$ such that $AB=\sum_{k=0}^{n-1}c_kP_{\lambda}^k$ so that $AB$ is also $\lambda$-circulant. In fact, it is true that
	\begin{equation*}
	c_{k}=\sum_{\substack{[i+j]_n=k\\i+j<n}}a_ib_j+\sum_{\substack{[i+j]_n=k\\i+j\geq n}}\lambda a_ib_j=\vct{x}_1\vct{y}_{k+1}
	\end{equation*}
for $k\in\{0,\ldots,n-1\}$, where $\vct{x}_i$ and $\vct{y}_i$ respectively denote the $i\nth$ row and column of $A$ and $B$ and $[i+j]_n$ denotes the smallest non-negative integer such that $[i+j]_n\equiv i+j\mod{n}$. From this, we can see that $\lambda$-circulant matrices commute multiplicatively and in fact the set of $\lambda$-circulant matrices over a commutative ring of fixed size is itself a commutative ring. Moreover, if $\lambda$ is a unit in $R$, then $A^T$ is $\lambda^{-1}$-circulant such that $A^T=a_0I_n+\lambda\sum_{i=1}^{n-1}a_{n-i}P_{\lambda^{-1}}^i$. It follows then that $AA^T$ is $\lambda$-circulant if and only if $\lambda$ is involutory in $R$, i.e. $\lambda^2=1$. 

Let $J_n$ be an $n\times n$ matrix over $R$ whose $(i,j)\nth$ entry is $1$ if $i+j=n+1$ and 0 if otherwise. Then $J_n$ is called the $n\times n$ exchange matrix and corresponds to the row-reversed (or column-reversed) version of $I_n$. We see that $J_n$ is both symmetric and involutory, i.e. $J_n=J_n^T$ and $J_n^2=I_n$. For any matrix $A\in R^{m\times n}$, premultiplying $A$ by $J_m$ and postmultiplying $A$ by $J_n$ inverts the order in which the rows and columns of $A$ appear, respectively. Namely, the $(i,j)\nth$ entries of $J_mA$ and $AJ_n$ are the $([1-i]_m,j)\nth$ and $(i,[1-j]_n)\nth$ entries of $A$, respectively. Note that $[i+j]_n$ corresponds to the $(i+1,j+1)\nth$ entry of the matrix $J_nV$ where $V=\cir{n-1,0,1,\ldots,n-2}$ for $i,j\in\{0,\ldots,n-1\}$.

For a $\lambda$-circulant matrix $A$ of size $n\times n$, it is easy to see that the sum of entries in the $i\nth$ row of $A$ is equal to the sum of entries in the $(n-i+1)\nth$ column of $A$. Thus, for $\vctg{\xi}=(\xi,\xi,\ldots,\xi)\in R^n$, we have $\vctg{\xi}A=\vctg{\xi}A^TJ_n$.

\section{The Construction}
In this section, we present our technique for constructing self-dual codes. We will hereafter always assume $R$ is a finite commutative Frobenius ring of characteristic 2.
	\begin{lemma}\label{lemma-1}
		Let 
		\begin{align*}
			X=\begin{pmatrix}
			    AC & B\\
			    B^TC & A^T
			\end{pmatrix},
		\end{align*}
		where $A=\pcir{\lambda}{\vct{a}}$, $B=\pcir{\lambda}{\vct{b}}$ and $C=\pcir{\mu}{\vct{c}}$ with $\vct{a},\vct{b},\vct{c}\in R^n$ and $\lambda,\mu\in R:\lambda^2=\mu^2=1$. If $CC^T=I_n$, then $XX^T=I_{2n}$ if and only if $AA^T+BB^T=I_n$.
	\end{lemma}
		\begin{proof}
			 Since $\lambda^2=1$ by assumption, we have that $A$ and $B$ as well as their transpositions all commute with one another multiplicatively. We have
				\begin{equation*}
				XX^T=
				\begin{pmatrix}
				    AC & B\\
				    B^TC & A^T
				\end{pmatrix}
				\begin{pmatrix}
				    C^TA^T & C^TB\\
				    B^T & A
				\end{pmatrix}=
				\begin{pmatrix}
				    x_{1,1} & x_{1,2}\\
				    x_{1,2}^T & x_{2,2}
				\end{pmatrix},
				\end{equation*}	
			where
				\begin{align*}
				x_{1,1}&=ACC^TA^T+BB^T,\\
				x_{1,2}&=ACC^TB+BA,\\
				x_{2,2}&=B^TCC^TB+A^TA
				\end{align*}
			and since $CC^T=I_n$ by assumption, we get
				\begin{align*}
				x_{1,1}&=AA^T+BB^T,\\
				x_{1,2}&=AB+BA=2AB=\vct{0},\\
				x_{2,2}&=B^TB+A^TA=AA^T+BB^T.
				\end{align*}

			Therefore, $XX^T=I_{2n}$ if and only if $AA^T+BB^T=I_n$.
		\end{proof}
	\begin{theorem}\label{theorem-1}
	Let $n\in\ZZ^+$ such that $n$ is odd and let
		\begin{align*}
			G=
	     	\begin{pmatrix}[c|c|cc]
				\vct{v} & \vct{0} & \xi_3 & \xi_4\\\midrule
				I_{2n} & X & \vct{v}^T & \vct{v}^T
			\end{pmatrix}
			,\quad\text{where }
			X=\begin{pmatrix}
			    AC & B\\
			    B^TC & A^T
			\end{pmatrix}
		\end{align*}
	where $\vct{v}=(\vctg{\xi}_1,\vctg{\xi}_2)$ with $\vctg{\xi}_i=(\xi_i,\xi_i,\ldots,\xi_i)\in R^n$ for $i\in\{1,2\}$ and $\xi_3,\xi_4\in R$ also with $A=\pcir{\lambda}{\vct{a}}$, $B=\pcir{\lambda}{\vct{b}}$ and $C=\pcir{\mu}{\vct{c}}$ for $\vct{a},\vct{b},\vct{c}\in R^n$ and $\lambda,\mu\in R:\lambda^2=\mu^2=1$. If $CC^T=I_n$, then $G$ is a generator matrix of a self-dual code of length $2(2n+1)$ if and only if
		\begin{empheq}[left=\empheqlbrace]{align*}
			AA^T+BB^T&=I_n,\\
			\sum_{i=1}^4\xi_i^2&=0,\\
			\xi_j(\xi_3+\xi_4+1)&=0\quad j\in\{1,2\},
		\end{empheq}
	and the free rank of $(\vctg{\xi}_1A+\vctg{\xi}_2B, \vctg{\xi}_1B+\vctg{\xi}_2A,\xi_3,\xi_4)$ is 1.
	\end{theorem}
		\begin{proof}
			First, let us determine the conditions required for $G$ to be a generator matrix of a self-orthogonal code. We have
				\begin{align*}
					GG^T&=
					\begin{pmatrix}[c|c|cc]
					\vct{v} & \vct{0} & \xi_3 & \xi_4\\\midrule
					I_{2n} & X & \vct{v}^T & \vct{v}^T
					\end{pmatrix}
					\begin{pmatrix}[c|c]
					\vct{v}^T & I_{2n}\\\midrule
					\vct{0} & X^T\\\midrule
					\xi_3 & \vct{v}\\
					\xi_4 & \vct{v}
					\end{pmatrix}=
					\begin{pmatrix}
						g_{1,1} & g_{1,2}\\
						g_{1,2}^T & g_{2,2}
					\end{pmatrix},
				\end{align*}
			where
				\begin{align*}
					g_{1,1}&=\vct{v}\vct{v}^T+\xi_3^2+\xi_4^2,\\
					g_{1,2}&=(\xi_3+\xi_4+1)\vct{v},\\
					g_{2,2}&=XX^T+I_{2n}+2\vct{v}^T\vct{v}
				\end{align*}
			so that $GG^T=\vct{0}$ if and only if $g_{1,1}=0$, $g_{1,2}=\vct{0}$ and $g_{2,2}=\vct{0}$. Since $R$ is of characteristic 2 and $n$ is odd, we have
				\begin{align*}
					g_{1,1}&=\vct{v}\vct{v}^T+\xi_3^2+\xi_4^2\\&=
					n(\xi_1^2+\xi_2^2)+\xi_3^2+\xi_4^2\\&=
					\xi_1^2+\xi_2^2+\xi_3^2+\xi_4^2,
				\end{align*}	
			so $g_{1,1}=0$ if and only if $\sum_{i=1}^4\xi_i^2=0$. We also have $2\vct{v}^T\vct{v}=\vct{0}$,
			so $g_{2,2}=\vct{0}$ if and only if $XX^T=I_{2n}$. Since $CC^T=I_n$ by assumption, it follows from Lemma \ref{lemma-1} that $XX^T=I_{2n}$ if and only if $AA^T+BB^T=I_n$. Finally , we see that $g_{1,2}=\vct{0}$ if and only if $\xi_j(\xi_3+\xi_4+1)=0$ for $j\in\{1,2\}$.
			
			Assume now that $G$ is a generator matrix of a self-orthogonal code. We need to prove that the free rank of $G$ is $2n+1$ if and only if the free rank of $(\vctg{\xi}_1A+\vctg{\xi}_2B,\vctg{\xi}_1B+\vctg{\xi}_2A,\xi_3,\xi_4)$ is 1. The free rank of $G$ is unchanged by elementary row (or column) operations and premultiplication (or postmultiplication) by an invertible matrix of appropriate size. Let $\tilde{G}=GM$ where
				\begin{align*}
					M=\begin{pmatrix}[c|c|cc]
						I_{2n} & X & \begin{matrix}\vct{v}^T & \vct{v}^T\end{matrix}\\\midrule
						\vct{0} & I_{2n} & \vct{0}\\\midrule
						\vct{0} & \vct{0} & I_2
					\end{pmatrix}.
				\end{align*}
			
			Let $\rank{}$ denote the free rank of a matrix over $R$. It is clear that $M$ is invertible and hence $\rank{\tilde{G}}=\rank{G}$. We have that
				\begin{align*}
					GM&=
					\begin{pmatrix}[c|c|cc]
						\vct{v} & \vct{0} & \xi_3 & \xi_4\\\midrule
						I_{2n} & X & \vct{v}^T & \vct{v}^T
					\end{pmatrix}
					\begin{pmatrix}[c|c|cc]
						I_{2n} & X & \begin{matrix} \vct{v}^T & \vct{v}^T\end{matrix}\\\midrule
						\vct{0} & I_{2n} & \vct{0}\\\midrule
						\vct{0} & \vct{0} & I_2
					\end{pmatrix}\\&=
					\begin{pmatrix}
						\vct{v} & \vct{v}X & \vct{v}(\vct{v}^T,\vct{v}^T)+(\xi_3,\xi_4)\\
						I_{2n} & \vct{0} & \vct{0}
					\end{pmatrix}\\&=
					\begin{pmatrix}
						\vct{v} & \vct{v}X & (\vct{v}\vct{v}^T+\xi_3,\vct{v}\vct{v}^T+\xi_4)\\
						I_{2n} & \vct{0} & \vct{0}
					\end{pmatrix}.
				\end{align*}
			
			Let $r=\rank{(\vct{v}X,\vct{v}\vct{v}^T+\xi_3,\vct{v}\vct{v}^T+\xi_4)}$. Then $\rank{\tilde{G}}=2n+1$ if and only if $r=1$. We see that
				\begin{align*}
					\vct{v}X&=
					(\vctg{\xi}_1,\vctg{\xi}_2)
					\begin{pmatrix}
						AC & B\\
						B^TC & A^T
					\end{pmatrix}\\&=
					((\vctg{\xi}_1A+\vctg{\xi}_2B^T)C,\vctg{\xi}_1B+\vctg{\xi}_2A^T)
				\end{align*}
			and
				\begin{align*}
					(\vct{v}\vct{v}^T+\xi_3,\vct{v}\vct{v}^T+\xi_4)=
					(\xi_3+\xi_1^2+\xi_2^2,\xi_4+\xi_1^2+\xi_2^2).
				\end{align*}	
			
			Since $G$ is a generator matrix of a self-orthogonal code, we have $\sum_{i=1}^4\xi_i^2=0$ so that $\xi_1^2+\xi_2^2=\xi_3^2+\xi_4^2$. By elementary column operations we obtain
				\begin{align*}
					r&=\rank{(\vct{v}X,\vct{v}\vct{v}^T+\xi_3,\vct{v}\vct{v}^T+\xi_4)}\\&=
					\rank{(\vct{v}X,\xi_3+\xi_1^2+\xi_2^2,\xi_4+\xi_1^2+\xi_2^2)}\\&=
					\rank{(\vct{v}X,\xi_3+\xi_3^2+\xi_4^2,\xi_4+\xi_3^2+\xi_4^2)}\\&=
					\rank{(\vct{v}X,\xi_3+\xi_3^2+\xi_4^2,\xi_3+\xi_4)}\\&=
					\rank{(\vct{v}X,\xi_3+\xi_3^2+\xi_4^2+(\xi_3+\xi_4)^2,\xi_3+\xi_4)}\\&=
					\rank{(\vct{v}X,\xi_3,\xi_3+\xi_4)}\\&=
					\rank{(\vct{v}X,\xi_3,\xi_4)}.
				\end{align*}
			
			We also have that $CC^T=I_n$ so that $C$ is invertible. We also recall that $\vctg{\xi}_2A=\vctg{\xi}_2A^TJ_n$ so that $\vctg{\xi}_2AJ_n=\vctg{\xi}_2A^T$ and likewise $\vctg{\xi}_2BJ_n=\vctg{\xi}_2B^T$. Thus, we get
				\begin{align*}
					r&=\rank{(\vct{v}X,\xi_3,\xi_4)}\\&=
					\rank{((\vctg{\xi}_1A+\vctg{\xi}_2B^T)C,\vctg{\xi}_1B+\vctg{\xi}_2A^T,\xi_3,\xi_4)}\\&=
					\rank{(\vctg{\xi}_1A+\vctg{\xi}_2BJ_n,\vctg{\xi}_1B+\vctg{\xi}_2AJ_n,\xi_3,\xi_4)}\\&=
					\rank{(\vctg{\xi}_1A+\vctg{\xi}_2B,\vctg{\xi}_1B+\vctg{\xi}_2A,\xi_3,\xi_4)}
				\end{align*}
			and so $\rank{G}=2n+1$ if and only if $\rank{(\vctg{\xi}_1A+\vctg{\xi}_2B, \vctg{\xi}_1B+\vctg{\xi}_2A,\xi_3,\xi_4)}=1$.
		\end{proof}

\section{Results}

In this section, we apply Theorem \ref{theorem-1} to obtain many new best known, optimal and extremal binary self-dual codes. In particular, we obtain 8 singly-even $[56,28,10]$ codes, one $[78,39,14]$ code, 3 $[92,46,16]$ codes and 43 $[94,47,16]$ codes.

We also apply the following well-known technique for constructing self-dual codes referred to as the building-up construction.
	\begin{theorem}\label{theorem-2}\textup{(\cite{R-068})}
	Let $R$ be a commutative Frobenius ring. Let $G'$ be a generator matrix of a self-dual code $\mathcal{C}'$ of length $2n$ over $R$ and let $\vct{r}_i$ denote the $i\nth$ row of $G'$. Let $\varepsilon\in R:\varepsilon^2=-1$, $\vctg{\updelta}\in R^{2n}:\langle\vctg{\updelta},\vctg{\updelta}\rangle=-1$ and $\gamma_i=\langle\vct{r}_i,\vctg{\updelta}\rangle$ for $i\in\{1,\ldots,n\}$. Then the matrix
		\begin{equation*}
		G=\begin{pmatrix}[cc|c]
		1 & 0 & \vctg{\updelta}\\\hline
		-\gamma_1 & \varepsilon\gamma_1 & \vct{r}_1\\
		-\gamma_2 & \varepsilon\gamma_2 & \vct{r}_2\\
		\vdots & \vdots & \vdots\\
		-\gamma_n & \varepsilon\gamma_n & \vct{r}_n
		\end{pmatrix}
		\end{equation*}
	is a generator matrix of a self-dual code of length $2(n+1)$ over $R$.
	\end{theorem}

By applying Theorem \ref{theorem-2}, we obtain 3 singly-even $[56,28,10]$ codes.

	\begin{remark}\label{remark-1}
	Two binary self-dual codes of length $2n$ are said to be neighbours if their intersection has dimension $n-1$. Let $\mathcal{C}^*$ be a binary self-dual code of length $2n$ and let $\vct{x}\in\FF_2^{2n}\setminus \mathcal{C}^*$. Then $\mathcal{C}=\langle\langle\vct{x}\rangle^{\perp}\cap\mathcal{C}^*,\vct{x}\rangle$ is a neighbour of $\mathcal{C}^*$, where $\langle\vct{x}\rangle$ denotes the code generated by $\vct{x}$.
	\end{remark}
	
Using Remark \ref{remark-1}, we obtain one new $[62,31,12]$ code as a neighbour of a $[62,31,12]$ code constructed by applying Theorem \ref{theorem-1}.

We conduct the search for these codes using MATLAB and Magma \cite{Magma} and determine their properties using Q-extension \cite{Q-extension} and Magma. In MATLAB, we employ an algorithm which randomly searches for the construction parameters that satisfy the necessary and sufficient conditions stated in Theorem \ref{theorem-1}. For such parameters, we then build the corresponding binary generator matrices and print them to text files. We then use Q-extension to read these text files and determine the minimum distance and partial weight enumerator of each corresponding code. Furthermore, we determine the automorphism
group order of each code using Magma. We implement a similar procedure for Theorem \ref{theorem-2}. Magma is used to search for neighbours as described in Remark \ref{remark-1}. A database of generator matrices of the new codes is given online at \cite{GMD}. The database is partitioned into text files (interpretable by Q-extension) corresponding to each code type. In these files, specific properties of the codes including the construction parameters, weight enumerator parameter values and automorphism group order are formatted as comments above the generator matrices. Partial weight enumerators of the codes are also formatted as comments below the generator matrices. Table \ref{table-1} gives the hexadecimal notation system we use to represent elements of $\FF_2+u\FF_2$, $\FF_2+u\FF_2+v\FF_2+uv\FF_2$ and $\FF_4+u\FF_4$. 

	\begin{table}
	\caption{Hexadecimal notation system for elements of $\FF_2+u\FF_2$, $\FF_2+u\FF_2+v\FF_2+uv\FF_2$ and $\FF_4+u\FF_4$.}\label{table-1}
	\centering
	\begin{adjustbox}{max width=\textwidth}
	\footnotesize
	\setlength{\tabcolsep}{4pt}
	\begin{tabular}{cccc}\midrule
	$\FF_2+u\FF_2$ & $\FF_2+u\FF_2+v\FF_2+uv\FF_2$ & $\FF_4+u\FF_4$ & Symbol\\\midrule
	$0$   & $0$        & $0$        & \texttt{0}\\
	$1$   & $1$        & $1$        & \texttt{1}\\
	$u$   & $u$        & $w$        & \texttt{2}\\
	$1+u$ & $1+u$      & $1+w$      & \texttt{3}\\
	$-$   & $v$        & $u$        & \texttt{4}\\
	$-$   & $1+v$      & $1+u$      & \texttt{5}\\
	$-$   & $u+v$      & $w+u$      & \texttt{6}\\
	$-$   & $1+u+v$    & $1+w+u$    & \texttt{7}\\
	$-$   & $uv$       & $wu$       & \texttt{8}\\
	$-$   & $1+uv$     & $1+wu$     & \texttt{9}\\
	$-$   & $u+uv$     & $w+wu$     & \texttt{A}\\
	$-$   & $1+u+uv$   & $1+w+wu$   & \texttt{B}\\
	$-$   & $v+uv$     & $u+wu$     & \texttt{C}\\
	$-$   & $1+v+uv$   & $1+u+wu$   & \texttt{D}\\
	$-$   & $u+v+uv$   & $w+u+wu$   & \texttt{E}\\
	$-$   & $1+u+v+uv$ & $1+w+u+wu$ & \texttt{F}\\\midrule
	\end{tabular}
	\end{adjustbox}
	\end{table}
	
\subsection{New Self-Dual Codes of Length 56}
	
The possible weight enumerators of a singly-even binary self-dual $[56,28,10]$ code are given in \cite{R-030} as
	\begin{align*}
	    W_{56,1}&=1+(308+4\alpha)x^{10}+(4246-8\alpha)x^{12}+\cdots,\\
	    W_{56,2}&=1+(308+4\alpha)x^{10}+(3990-8\alpha)x^{12}+\cdots,
	\end{align*}
where $\alpha\in\ZZ$. Previously known $\alpha$ values for weight enumerators $W_{56,1}$ and $W_{56,2}$ can be found online at \cite{WEPD} (see \cite{R-030,R-080,R-098,MRes,AMR1}).

We obtain 11 new best known singly-even binary self-dual codes of length 56 of which 6 have weight enumerator $W_{56,1}$ for
	\begin{enumerate}[label=]
        \item $\alpha\in\{-z:z=45,\lb47,\lb49,\lb50,\lb54,\lb55\}$
	\end{enumerate}	
and 5 have weight enumerator $W_{56,2}$ for
	\begin{enumerate}[label=]
	    \item $\alpha\in\{-z:z=50,\lb...,\lb53,\lb55\}$.
	\end{enumerate}	

Of the 11 new codes, 3 are constructed by first applying Theorem \ref{theorem-1} to obtain codes of length 54 over $\FF_2$ (Table \ref{table-54}) to which we then apply Theorem \ref{theorem-2} (Table \ref{table-56-F2}); 6 are constructed by applying Theorem \ref{theorem-1} over $\FF_2+u\FF_2+v\FF_2+uv\FF_2$ (Table \ref{table-56-F2uvuv}) and 2 are constructed by applying Theorem \ref{theorem-1} over $\FF_4+u\FF_4$ (Table \ref{table-56-F4u}).

	\begin{table}
	\caption{Code of length 54 over $\FF_2$ from Theorem \ref{theorem-1} to which we apply Theorem \ref{theorem-2} to obtain the codes in Table \ref{table-56-F2}, where $\vctg{\xi}=(\xi_1,\xi_2,\xi_3,\xi_4)$.}\label{table-54}
	\centering
	\begin{adjustbox}{max width=\textwidth}
	\footnotesize
	\setlength{\tabcolsep}{4pt}
	\begin{tabular}{ccccc}\midrule
	$\mathcal{C}_{54,i}'$ & $\vct{a}$ & $\vct{b}$ & $\vct{c}$ & $\vctg{\xi}$\\\midrule
	1 & \texttt{(0100111100101)} & \texttt{(1111101111010)} & \texttt{(1011001111110)} & \texttt{(1010)}
	\\\midrule
	\end{tabular}
	\end{adjustbox}
	\end{table}

	\begin{table}
	\caption{New singly-even binary self-dual $[56,28,10]$ codes from applying Theorem \ref{theorem-2} to $\mathcal{C}_{54,j}'$ as given in Table \ref{table-54} with $\vctg{\delta}=(\vct{0},\vctg{\delta}_0)$.}\label{table-56-F2}
	\centering
	\begin{adjustbox}{max width=\textwidth}
	\footnotesize
	\setlength{\tabcolsep}{4pt}
	\begin{tabular}{cccccc}\midrule
	$\mathcal{C}_{56,i}$ & $\mathcal{C}_{54,j}'$ & $\vctg{\updelta}_0$ & $W_{56,k}$ & $\alpha$ & $|\aut{\mathcal{C}_{56,i}}|$\\\midrule
	1 & 1 & \texttt{(000101100101100011111000101)} & 1 & $-55$ & $1$\\
	2 & 1 & \texttt{(001000100111010111011101110)} & 1 & $-47$ & $1$\\
	3 & 1 & \texttt{(110101100010101111001101100)} & 2 & $-50$ & $1$\\\midrule
	\end{tabular}
	\end{adjustbox}
	\end{table}
	
	\begin{table}
	\caption{New singly-even binary self-dual $[56,28,10]$ codes from Theorem \ref{theorem-1} over $\FF_2+u\FF_2+v\FF_2+uv\FF_2$, where $\vctg{\xi}=(\xi_1,\xi_2,\xi_3,\xi_4)$.}\label{table-56-F2uvuv}
	\centering
	\begin{adjustbox}{max width=\textwidth}
	\footnotesize
	\setlength{\tabcolsep}{4pt}
	\begin{tabular}{cccccccccc}\midrule
	$\mathcal{C}_{56,i}$ & $\lambda$ & $\mu$ & $\vct{a}$ & $\vct{b}$ & $\vct{c}$ & $\vctg{\xi}$ & $W_{56,j}$ & $\alpha$ & $|\aut{\mathcal{C}_{56,i}}|$\\\midrule
	4 & \texttt{1} & \texttt{9} & \texttt{(B03)} & \texttt{(39D)} & \texttt{(344)} & \texttt{(7EBA)} & 1 & $-54$ & $2^{2}\cdot 3$\\
	5 & \texttt{F} & \texttt{1} & \texttt{(331)} & \texttt{(8F9)} & \texttt{(EE3)} & \texttt{(DEBA)} & 1 & $-50$ & $2^{2}$\\
	6 & \texttt{1} & \texttt{3} & \texttt{(D52)} & \texttt{(F95)} & \texttt{(700)} & \texttt{(9EAB)} & 2 & $-55$ & $2^{2}\cdot 3$\\
	7 & \texttt{1} & \texttt{5} & \texttt{(DB9)} & \texttt{(D45)} & \texttt{(88D)} & \texttt{(D654)} & 2 & $-53$ & $2^{3}\cdot 3$\\
	8 & \texttt{1} & \texttt{1} & \texttt{(31F)} & \texttt{(54D)} & \texttt{(00F)} & \texttt{(D6DC)} & 2 & $-52$ & $2^{3}\cdot 3$\\
	9 & \texttt{1} & \texttt{F} & \texttt{(FA7)} & \texttt{(FB5)} & \texttt{(C47)} & \texttt{(76AB)} & 2 & $-51$ & $2^{2}$\\\midrule
	\end{tabular}
	\end{adjustbox}
	\end{table}

	\begin{table}
	\caption{New singly-even binary self-dual $[56,28,10]$ codes from Theorem \ref{theorem-1} over $\FF_4+u\FF_4$, where $\vctg{\xi}=(\xi_1,\xi_2,\xi_3,\xi_4)$.}\label{table-56-F4u}
	\centering
	\begin{adjustbox}{max width=\textwidth}
	\footnotesize
	\setlength{\tabcolsep}{4pt}
	\begin{tabular}{cccccccccc}\midrule
	$\mathcal{C}_{56,i}$ & $\lambda$ & $\mu$ & $\vct{a}$ & $\vct{b}$ & $\vct{c}$ & $\vctg{\xi}$ & $W_{56,j}$ & $\alpha$ & $|\aut{\mathcal{C}_{56,i}}|$\\\midrule
	10 & \texttt{1} & \texttt{1} & \texttt{(48D)} & \texttt{(5F2)} & \texttt{(CC9)} & \texttt{(6F67)} & 1 & $-49$ & $2\cdot 3$\\
	11 & \texttt{D} & \texttt{9} & \texttt{(B5D)} & \texttt{(D61)} & \texttt{(900)} & \texttt{(6F67)} & 1 & $-45$ & $2$\\\midrule
	\end{tabular}
	\end{adjustbox}
	\end{table}

\subsection{New Self-Dual Code of Length 62}

The possible weight enumerators of a binary self-dual $[62,31,12]$ code are given in \cite{R-191} as
	\begin{align*}
	    W_{62,1}&=1+2308x^{12}+23767x^{14}+\cdots,\\
	    W_{62,2}&=1+(1860+32\alpha)x^{12}+(28055-160\alpha)x^{14}+\cdots,
	\end{align*}
where $\alpha\in\ZZ$. Previously known $\alpha$ values for weight enumerator $W_{62,2}$ can be found online at \cite{WEPD} (see \cite{R-192,R-191,R-031,R-193,R-125}).

We obtain one new extremal binary self-dual codes of length 62 which has weight enumerator $W_{62,2}$ for
	\begin{enumerate}[label=]
	    \item $\alpha=2$.
	\end{enumerate}	

The new code is constructed by first applying Theorem \ref{theorem-1} to obtain a code of length 62 over $\FF_2$ (Table \ref{table-62-1}) and then searching for neighbours of this code using Remark \ref{remark-1} (Table \ref{table-62-2}).

	\begin{table}
	\caption{Code of length 62 over $\FF_2$ from Theorem \ref{theorem-1} to which we apply Remark \ref{remark-1} to obtain the code in Table \ref{table-62-2}, where $\vctg{\xi}=(\xi_1,\xi_2,\xi_3,\xi_4)$.}\label{table-62-1}
	\centering
	\begin{adjustbox}{max width=\textwidth}
	\footnotesize
	\setlength{\tabcolsep}{4pt}
	\begin{tabular}{ccccc}\midrule
	$\mathcal{C}_{62,i}^*$ & $\vct{a}$ & $\vct{b}$ & $\vct{c}$ & $\vctg{\xi}$\\\midrule
	1 & \texttt{(000000100100101)} & \texttt{(000011101110111)} & \texttt{(100000000000000)} & \texttt{(0110)}\\\midrule
	\end{tabular}
	\end{adjustbox}
	\end{table}

	\begin{table}
	\caption{New binary self-dual $[62,31,12]$ code from searching for neighbours of $\mathcal{C}_{62,j}^*$ as given in Table \ref{table-62-1} using Remark \ref{remark-1} with $\vct{x}=(\vct{0},\vct{x}_0)$.}\label{table-62-2}
	\centering
	\begin{adjustbox}{max width=\textwidth}
	\footnotesize
	\setlength{\tabcolsep}{4pt}
	\begin{tabular}{cccccc}\midrule
	$\mathcal{C}_{62,i}$ & $\mathcal{C}_{62,j}^*$ & $\vct{x}_0$ & $W_{62,k}$ & $\alpha$ & $|\aut{\mathcal{C}_{62,i}}|$\\\midrule
	1 & 1 & \texttt{(1111001101001110100110000100110)} & 2 & $2$ & $2^{2}$\\\midrule
	\end{tabular}
	\end{adjustbox}
	\end{table}

\subsection{New Self-Dual Code of Length 78}

The possible weight enumerators of a binary self-dual $[78,39,14]$ code are given in \cite{R-044,R-033} as
	\begin{align*}
	    W_{78,1}&=1+(3705+8\alpha)x^{14}+(62244-24\alpha+512\beta)x^{16}\\&\quad+(774592-64\alpha-4608\beta)x^{18}+\cdots,\\
	    W_{78,2}&=1+(3705+8\alpha)x^{14}+(71460-24\alpha)x^{16}\\&\quad+(658880-64\alpha)x^{18}+\cdots,
	\end{align*}
where $\alpha,\beta\in\ZZ$. Previously known $(\alpha,\beta)$ values for weight enumerator $W_{78,1}$ can be found online at \cite{WEPD} (see \cite{R-044,R-173,R-033,R-026,R-049,R-086,R-098}).

We obtain one new optimal binary self-dual codes of length 78 which has weight enumerator $W_{78,1}$ for
	\begin{enumerate}[label=]
	    \item $\beta=0$ and $\alpha=-76$.
	\end{enumerate}	

The new code is constructed by applying Theorem \ref{theorem-1} over $\FF_2$ (Table \ref{table-78}).

	\begin{table}
	\caption{New binary self-dual $[78,39,14]$ code from Theorem \ref{theorem-1} over $\FF_2$, where $\vctg{\xi}=(\xi_1,\xi_2,\xi_3,\xi_4)$.}\label{table-78}
	\centering
	\begin{adjustbox}{max width=\textwidth}
	\footnotesize
	\setlength{\tabcolsep}{4pt}
	\begin{tabular}{ccccc}\midrule
	$\mathcal{C}_{78,i}$ & $\vct{a}$ & $\vct{b}$ & $\vct{c}$ & $\vctg{\xi}$ \\\midrule
	1 & \texttt{(0100101010100001000)} & \texttt{(1111101101011010000)} & \texttt{(0010101111111101101)} & \texttt{(0101)}\\\midrule
	\end{tabular}
	\end{adjustbox}
	\end{table}

	\addtocounter{table}{-1}
	\begin{table}
	\caption{(continued)}
	\centering
	\begin{adjustbox}{max width=\textwidth}
	\footnotesize
	\setlength{\tabcolsep}{4pt}
	\begin{tabular}{ccccc}\midrule
	$\mathcal{C}_{78,i}$ & $W_{78,j}$ & $\alpha$ & $\beta$ & $|\aut{\mathcal{C}_{78,i}}|$\\\midrule
	1 & 1 & $-76$ & $0$ & $19$\\\midrule
	\end{tabular}
	\end{adjustbox}
	\end{table}

\subsection{New Self-Dual Codes of Length 92}

The possible weight enumerators of a binary self-dual $[92,46,16]$ code are given in \cite{R-044} as
	\begin{align*}
	    W_{92,1}&=1+(4692+4\alpha)x^{16}+(174800-8\alpha+256\beta)x^{18}\\&\quad+(2425488-52\alpha-2048\beta)x^{20}+\cdots,\\
	    W_{92,2}&=1+(4692+4\alpha)x^{16}+(174800-8\alpha+256\beta)x^{18}\\&\quad+(2441872-52\alpha-2048\beta)x^{20}+\cdots,\\
	    W_{92,3}&=1+(4692+4\alpha)x^{16}+(121296-8\alpha)x^{18}\\&\quad+(3213968-52\alpha)x^{20}+\cdots,
	\end{align*}
where $\alpha,\beta\in\ZZ$. Previously known $(\alpha,\beta)$ values for weight enumerator $W_{92,1}$ can be found online at \cite{WEPD} (see \cite{R-134,R-118,R-136,R-137,R-096,AMR1}).

We obtain 3 new extremal binary self-dual codes of length 92 which have weight enumerator $W_{92,1}$ for
	\begin{enumerate}[label=]
	    \item $\beta=0$ and $\alpha\in\{807,862,1038\}$.
	\end{enumerate}	

The new codes are constructed by applying Theorem \ref{theorem-1} over $\FF_2+u\FF_2$ (Table \ref{table-92}).

	\begin{table}
	\caption{New binary self-dual $[92,46,16]$ codes from Theorem \ref{theorem-1} over $\FF_2+u\FF_2$, where $\vctg{\xi}=(\xi_1,\xi_2,\xi_3,\xi_4)$.}\label{table-92}
	\centering
	\begin{adjustbox}{max width=\textwidth}
	\footnotesize
	\setlength{\tabcolsep}{4pt}
	\begin{tabular}{ccccccccccc}\midrule
	$\mathcal{C}_{92,i}$ & $\lambda$ & $\mu$ & $\vct{a}$ & $\vct{b}$ & $\vct{c}$ & $\vctg{\xi}$ & $W_{92,j}$ & $\alpha$ & $\beta$ & $|\aut{\mathcal{C}_{92,i}}|$\\\midrule
	1 & \texttt{1} & \texttt{1} & \texttt{(02223003031)} & \texttt{(02321323010)} & \texttt{(22232222222)} & \texttt{(2301)} & 1 & $807$ & $0$ & $2\cdot 11$\\
	2 & \texttt{1} & \texttt{1} & \texttt{(22101022322)} & \texttt{(33123310002)} & \texttt{(02101130201)} & \texttt{(1201)} & 1 & $862$ & $0$ & $2\cdot 11$\\
	3 & \texttt{1} & \texttt{1} & \texttt{(21200233200)} & \texttt{(31233320122)} & \texttt{(22232222222)} & \texttt{(2310)} & 1 & $1038$ & $0$ & $2^{2}\cdot 11$\\\midrule
	\end{tabular}
	\end{adjustbox}
	\end{table}

\subsection{New Self-Dual Codes of Length 94}

The possible weight enumerators of a binary self-dual $[94,47,16]$ code are given in \cite{R-184} as
	\begin{align*}
	    W_{94,1}&=1+2\alpha x^{16}+(134044-2\alpha+128\beta)x^{18}\\&\quad+(2010660-30\alpha-896\beta)x^{20}+\cdots,\\
	    W_{94,2}&=1+2\alpha x^{16}+(134044-2\alpha+128\beta)x^{18}\\&\quad+(2018852-30\alpha-896\beta)x^{20}+\cdots,\\
	    W_{94,3}&=1+2\alpha x^{16}+(134044-2\alpha+128\beta)x^{18}\\&\quad+(2190884-30\alpha-896\beta)x^{20}+\cdots,
	\end{align*}
where $\alpha,\beta\in\ZZ$. Previously known $(\alpha,\beta)$ values for weight enumerator $W_{94,1}$ can be found online at \cite{WEPD} (see \cite{R-184}).

We obtain 43 new best known binary self-dual codes of length 94 which have weight enumerator $W_{94,1}$ for
	\begin{enumerate}[label=]
	    \item $\beta=-69$ and $\alpha\in\{46z:z=78\}$;
	    \item $\beta=-46$ and $\alpha\in\{46z:z=69,\lb70,\lb76,\lb78,\lb79\}$;
	    \item $\beta=-23$ and $\alpha\in\{46z:z=62,\lb63,\lb65,\lb...,\lb79\}$;
	    \item $\beta=0$ and $\alpha\in\{46z:z=57,\lb...,\lb65,\lb67,\lb...,\lb75,\lb77,\lb79\}$.
	\end{enumerate}	

The new codes are constructed by applying Theorem \ref{theorem-1} over $\FF_2$ (Table \ref{table-94}).

	\begin{table}
	\caption{New binary self-dual $[94,47,16]$ codes from Theorem \ref{theorem-1} over $\FF_2$, where $\vctg{\xi}=(\xi_1,\xi_2,\xi_3,\xi_4)$.}\label{table-94}
	\centering
	\begin{adjustbox}{max width=\textwidth}
	\footnotesize
	\setlength{\tabcolsep}{4pt}
	\begin{tabular}{ccccc}\midrule
	$\mathcal{C}_{94,i}$ & $\vct{a}$ & $\vct{b}$ & $\vct{c}$ & $\vctg{\xi}$ \\\midrule
	1 & \texttt{(10001010001010111110101)} & \texttt{(01000001011100010110100)} & \texttt{(00111001110111000101110)} & \texttt{(1001)}\\
	2 & \texttt{(10011110100001000100100)} & \texttt{(00011011100101010111010)} & \texttt{(00111010111010011111000)} & \texttt{(0110)}\\
	3 & \texttt{(00101010100010100100010)} & \texttt{(11100001010001000001101)} & \texttt{(11111111000100100011010)} & \texttt{(1010)}\\
	4 & \texttt{(00000100011100110110010)} & \texttt{(00011111001000001101111)} & \texttt{(00010101100111100100000)} & \texttt{(0110)}\\
	5 & \texttt{(10011000101101011101110)} & \texttt{(10110011100000101011110)} & \texttt{(11100010011010111010101)} & \texttt{(0110)}\\
	6 & \texttt{(01000001110110010100101)} & \texttt{(11100010110101000111000)} & \texttt{(10010110110010111001011)} & \texttt{(1001)}\\
	7 & \texttt{(11100100101101110010101)} & \texttt{(11001101011001100101010)} & \texttt{(00100111000000010110011)} & \texttt{(0101)}\\
	8 & \texttt{(01001110101010101000110)} & \texttt{(01110011110110110001110)} & \texttt{(11010001000011101111110)} & \texttt{(1001)}\\
	9 & \texttt{(01110111111000111010001)} & \texttt{(11000000111010011100011)} & \texttt{(10001100101100101001000)} & \texttt{(1001)}\\
	10 & \texttt{(00011101100001100010000)} & \texttt{(10011101010000001110000)} & \texttt{(10011011010001111100011)} & \texttt{(0101)}\\
	11 & \texttt{(11111111110101011110101)} & \texttt{(11010011100101011111101)} & \texttt{(10001011101100011111001)} & \texttt{(0110)}\\
	12 & \texttt{(01000100110100111011001)} & \texttt{(10110111110000110101101)} & \texttt{(00001011111100101101110)} & \texttt{(1010)}\\
	13 & \texttt{(10010111000100001000011)} & \texttt{(01111100011110001100100)} & \texttt{(10101011110110000101101)} & \texttt{(0101)}\\
	14 & \texttt{(11100001111011111011011)} & \texttt{(00101000010110010100011)} & \texttt{(10101110101100100011110)} & \texttt{(1010)}\\
	15 & \texttt{(00111000011101010101010)} & \texttt{(01000111100000001111001)} & \texttt{(01011011110011100111000)} & \texttt{(1010)}\\
	16 & \texttt{(10110010010001101000110)} & \texttt{(01100011101110101000100)} & \texttt{(10111110010111010100100)} & \texttt{(1010)}\\
	17 & \texttt{(01111110110101010010100)} & \texttt{(00100010010010101110000)} & \texttt{(10000011010000001011011)} & \texttt{(0110)}\\
	18 & \texttt{(10011001010010000100000)} & \texttt{(10110110000111011111001)} & \texttt{(11001011101011100000111)} & \texttt{(0110)}\\
	19 & \texttt{(10111110111001000111001)} & \texttt{(00011100011101100010101)} & \texttt{(10001110111001110001110)} & \texttt{(0101)}\\
	20 & \texttt{(11100000011100111000111)} & \texttt{(01011101000011000010100)} & \texttt{(10001001101011101010111)} & \texttt{(0101)}\\
	21 & \texttt{(11000100000010010011001)} & \texttt{(10011111110011000001011)} & \texttt{(00010110110011100011111)} & \texttt{(1001)}\\
	22 & \texttt{(01000000000001111111111)} & \texttt{(11101101110111111011110)} & \texttt{(00111110100100101011101)} & \texttt{(1010)}\\
	23 & \texttt{(11101110101111101010001)} & \texttt{(00111000100001100101110)} & \texttt{(11011010101111011000010)} & \texttt{(1010)}\\
	24 & \texttt{(10111110010111110101001)} & \texttt{(00001100100001100110111)} & \texttt{(00010001001010011010011)} & \texttt{(0110)}\\
	25 & \texttt{(01101011111100100101100)} & \texttt{(10011111010101001011000)} & \texttt{(10100111101000111111000)} & \texttt{(0110)}\\
	26 & \texttt{(00110001110000010000100)} & \texttt{(01000110110110010010001)} & \texttt{(01001110000110001011000)} & \texttt{(1010)}\\
	27 & \texttt{(11111100110100101000011)} & \texttt{(01100110011000101111010)} & \texttt{(10000011111011010010111)} & \texttt{(1001)}\\
	28 & \texttt{(10111000111101100010111)} & \texttt{(10111011110011100110110)} & \texttt{(10101010001001000001101)} & \texttt{(1001)}\\
	29 & \texttt{(01110010100111101110101)} & \texttt{(01101110001001100110001)} & \texttt{(00110100000010110111000)} & \texttt{(0101)}\\
	30 & \texttt{(01101001100010111101111)} & \texttt{(11101000101110111111001)} & \texttt{(01101001110110011101001)} & \texttt{(1010)}\\
	31 & \texttt{(01100100110000110010010)} & \texttt{(01110111110111011110010)} & \texttt{(01011011010000110111101)} & \texttt{(0110)}\\
	32 & \texttt{(00111110101010110101110)} & \texttt{(11010100110110100100100)} & \texttt{(01110101000101100111101)} & \texttt{(1010)}\\
	33 & \texttt{(11010001010100000110111)} & \texttt{(00010111100111001101111)} & \texttt{(10010011000011011101111)} & \texttt{(0110)}\\
	34 & \texttt{(00011101000001110100111)} & \texttt{(00110000101010111100001)} & \texttt{(00110000111001000011010)} & \texttt{(1010)}\\
	35 & \texttt{(11010110011101000101011)} & \texttt{(01001110111001101001100)} & \texttt{(11100011000011101110110)} & \texttt{(0101)}\\
	36 & \texttt{(10010001001001011011010)} & \texttt{(10100110101000001110110)} & \texttt{(00000011101111110001111)} & \texttt{(1001)}\\
	37 & \texttt{(00010011001011000010101)} & \texttt{(01100100010011101111010)} & \texttt{(00111111011100000011110)} & \texttt{(0110)}\\
	38 & \texttt{(01100001111100111001010)} & \texttt{(10101111000011101100101)} & \texttt{(11111010010101100010110)} & \texttt{(0110)}\\
	39 & \texttt{(11001101110110000000111)} & \texttt{(11101011010011000101011)} & \texttt{(11100011011101000110011)} & \texttt{(0110)}\\
	40 & \texttt{(10111110110110010101000)} & \texttt{(01110001000111110000111)} & \texttt{(00101101110000010111111)} & \texttt{(1010)}\\
	41 & \texttt{(00000001001010101111101)} & \texttt{(01011000110110010010011)} & \texttt{(01111011101101010001001)} & \texttt{(1010)}\\
	42 & \texttt{(01110110010010100011011)} & \texttt{(01000011111001101101011)} & \texttt{(01000000101101110000011)} & \texttt{(0101)}\\
	43 & \texttt{(11101100100110110001001)} & \texttt{(01101011011000000000011)} & \texttt{(01000000000000000000000)} & \texttt{(1010)}\\\midrule
	\end{tabular}
	\end{adjustbox}
	\end{table}

	\addtocounter{table}{-1}
	\begin{table}
	\caption{(continued)}
	\centering
	\begin{adjustbox}{max width=\textwidth}
	\footnotesize
	\setlength{\tabcolsep}{4pt}
	\begin{tabular}{ccccc}\midrule
	$\mathcal{C}_{94,i}$ & $W_{94,j}$ & $\alpha$ & $\beta$ & $|\aut{\mathcal{C}_{94,i}}|$ \\\midrule
	1 & 1 & $3588$ & $-69$ & $23$\\
	2 & 1 & $3174$ & $-46$ & $23$\\
	3 & 1 & $3220$ & $-46$ & $23$\\
	4 & 1 & $3496$ & $-46$ & $23$\\
	5 & 1 & $3588$ & $-46$ & $23$\\
	6 & 1 & $3634$ & $-46$ & $23$\\
	7 & 1 & $2852$ & $-23$ & $23$\\
	8 & 1 & $2898$ & $-23$ & $23$\\
	9 & 1 & $2990$ & $-23$ & $23$\\
	10 & 1 & $3036$ & $-23$ & $23$\\
	11 & 1 & $3082$ & $-23$ & $23$\\
	12 & 1 & $3128$ & $-23$ & $23$\\
	13 & 1 & $3174$ & $-23$ & $23$\\
	14 & 1 & $3220$ & $-23$ & $23$\\
	15 & 1 & $3266$ & $-23$ & $23$\\
	16 & 1 & $3312$ & $-23$ & $23$\\
	17 & 1 & $3358$ & $-23$ & $23$\\
	18 & 1 & $3404$ & $-23$ & $23$\\
	19 & 1 & $3450$ & $-23$ & $23$\\
	20 & 1 & $3496$ & $-23$ & $23$\\
	21 & 1 & $3542$ & $-23$ & $23$\\
	22 & 1 & $3588$ & $-23$ & $23$\\
	23 & 1 & $3634$ & $-23$ & $23$\\
	24 & 1 & $2622$ & $0$ & $23$\\
	25 & 1 & $2668$ & $0$ & $23$\\
	26 & 1 & $2714$ & $0$ & $23$\\
	27 & 1 & $2760$ & $0$ & $23$\\
	28 & 1 & $2806$ & $0$ & $23$\\
	29 & 1 & $2852$ & $0$ & $23$\\
	30 & 1 & $2898$ & $0$ & $23$\\
	31 & 1 & $2944$ & $0$ & $23$\\
	32 & 1 & $2990$ & $0$ & $23$\\
	33 & 1 & $3082$ & $0$ & $23$\\
	34 & 1 & $3128$ & $0$ & $23$\\
	35 & 1 & $3174$ & $0$ & $23$\\
	36 & 1 & $3220$ & $0$ & $23$\\
	37 & 1 & $3266$ & $0$ & $23$\\
	38 & 1 & $3312$ & $0$ & $23$\\
	39 & 1 & $3358$ & $0$ & $23$\\
	40 & 1 & $3404$ & $0$ & $23$\\
	41 & 1 & $3450$ & $0$ & $23$\\
	42 & 1 & $3542$ & $0$ & $23$\\
	43 & 1 & $3634$ & $0$ & $2\cdot 23$\\\midrule
	\end{tabular}
	\end{adjustbox}
	\end{table}

\section{Conclusion}

In this work, we presented a new bordered construction for self-dual codes. We gave the necessary conditions for our method to produce self-dual codes over a finite commutative Frobenius ring of characteristic 2. Using our new bordered construction together with the well-known building-up and neighbour constructions, we were able to construct the following new singly-even binary self-dual codes:

\begin{enumerate}
\item[]\textbf{Codes of length 56:} We were able to construct new singly-even binary self-dual $[56,28,10]$ codes which have weight enumerator $W_{56,1}$ for:
	\begin{align*}
		\alpha\in\{-z:z=45,\lb47,\lb49,\lb50,\lb54,\lb55\}
	\end{align*}
and weight enumerator $W_{56,2}$ for:
	\begin{align*}
		\alpha\in\{-z:z=50,\lb...,\lb53,\lb55\}.
	\end{align*}

\item[]\textbf{Code of length 62:} We were able to construct a new binary self-dual $[62,31,12]$ code which has weight enumerator $W_{62,2}$ for:
	\begin{align*}
		\alpha=2.
	\end{align*}

\item[]\textbf{Code of length 78:} We were able to construct a new binary self-dual $[78,39,14]$ code which has weight enumerator $W_{78,1}$ for:
	\begin{align*}
		\beta=0 \ \text{and} \ \alpha=-76.
	\end{align*}

\item[]\textbf{Codes of length 92:} We were able to construct new binary self-dual $[92,46,16]$ codes which have weight enumerator $W_{92,1}$ for:
	\begin{align*}
		\beta=0 \ \text{and} \ \alpha\in\{807,862,1038\}.
	\end{align*}

\item[]\textbf{Codes of length 94:} We were able to construct new binary self-dual $[94,47,16]$ codes which have weight enumerator $W_{94,1}$ for:
	\begin{align*}
		&\beta=-69 \ \text{and} \ \alpha\in\{46z:z=78\},\\
		&\beta=-46 \ \text{and} \ \alpha\in\{46z:z=69,\lb70,\lb76,\lb78,\lb79\},\\
		&\beta=-23 \ \text{and} \ \alpha\in\{46z:z=62,\lb63,\lb65,\lb...,\lb79\},\\
		&\beta=0 \ \text{and} \ \alpha\in\{46z:z=57,\lb...,\lb65,\lb67,\lb...,\lb75,\lb77,\lb79\}.
	\end{align*}
\end{enumerate}

A suggestion for future work would be to consider our bordered construction with the $\lambda$-circulant matrices $A, B$ and $C$ being replaced with some other matrices, for example, matrices that come from group rings. This would give one many more possible choices for the matrices $A, B$ and $C$ which could lead to finding more new binary self-dual codes of different lengths.

\bibliographystyle{plainnat}
\bibliography{paper4}
\end{document}